\documentclass[preprint,12pt]{elsarticle} 

\usepackage[english]{babel} 
\usepackage{amssymb}
\usepackage[a4paper]{geometry}
\usepackage{graphicx}
\usepackage{microtype}
\usepackage{siunitx}
\usepackage{booktabs}
\usepackage[colorlinks=false, pdfborder={0 0 0}]{hyperref}
\usepackage{cleveref}

\usepackage{amsfonts}
\usepackage{verbatim}

\newcommand{\st}{\medskip\noindent}
\def\ni{\noindent}

\def\beq{\begin{equation}}
\def\eeq#1{\label{#1}\end{equation}}
\long\def\COMMENT#1\ENDCOMMENT{\message{(Commented text...)}\par}

\def\and{ \ \wedge}

\newcommand{\calw}{\mathcal{L}}

\newcommand{\G}{G\"odel}
\newcommand{\IP}{{\bf IP}}
\newcommand{\EP}{{\bf EP}}

\newcommand{\HA}{{\bf HA}}
\newcommand{\EA}{{\bf EA}}
\newcommand{\INT}{{\it Int}} 
\newcommand{\EPI}{{\it Epi}} 

\newcommand{\I}[1]{\relax\ifmmode\mbox{\it#1}\else{\it#1}\fi}

\newtheorem{theorem}{Theorem}
\newtheorem{proof}{Proof}

\journal{Annals of Pure and Applied Logic}

\begin{document}
\begin{frontmatter}
\title{A Note on Flagg and Friedman's Epistemic and Intuitionistic Formal Systems}

\author{Alessandro Provetti,  Andrea Zucchellini}
\address{Dipartimento di Matematica e Informatica\\
Universit\`a degli Studi di Messina\\
I-98166 Messina, Italy
}



\vspace{-1cm}
\begin{abstract}
We report our findings on the properties of Flagg and Friedman's translation from Epistemic into Intuitionistic logic, which was proposed as the basis of a comprehensive proof method for the faithfulness of the \G~translation. 
We focus on the propositional case and raise the issue of the admissibility of the translated necessitation rule. 
Then, we contribute to Flagg and Friedman's program by giving an explicit proof of the soundness of their translation.
\end{abstract}

\end{frontmatter} 

\section{Introduction}\label{sec:intro} 
In their work \textit{Epistemic and Intuitionistic Formal Systems} (1986), Flagg and Friedman gave a new proof--as an alternative to Goodman's \cite{Goo84a}--for the faithfulness of the \G~translation $(\cdot)^T$ from intuitionistic logic into epistemic modal logic \textbf{S4} once $(\cdot)^T$ is applied to Heyting arithmetic. 
They aimed at a comprehensive proof method, to be deployed not only for arithmetic but for various formal systems whose intuitionistic and epistemic versions are connected by a \textit{faithful} translation: 
 
\[ 
\vdash_\EPI A^T \Longrightarrow~\vdash_\INT A,
\]

\ni
for any intuitionistic formula $A$. 
To this end, they proposed a class of embeddings of epistemic into intuitionistic logic, which are, in a certain sense, inverse to $(\cdot)^T.$ 

In order to prove the faithfulness of the \G~translation, Flagg and Friedman claim that their translation%
\footnote{Notation and conventions will adhere as much as possible to Flagg and Friedman's \cite{FlaFri86}.}%
, $(\cdot)^{(E)}_\Gamma$, is {\em sound} in both the propositional and the predicate calculus, and in arithmetic: i.e. for any epistemic formula $B$

\[
\vdash_\EPI B \Longrightarrow~\vdash_\INT B^{(E)}_\Gamma,
\] 

\ni
where $\Gamma$ is an arbitrary set of intuitionistic formulae and $E$ is a chosen formula from $\Gamma.$

In this article we discuss some issues about the proof of soundness of the generic Flagg-and-Friedman translation $(\cdot)^{(E)}_\Gamma$ 
from Epistemic Propositional logic, shortly \EP~(i.e. the modal logic \textbf{S4}), into Intuitionistic Propositional logic, \IP.
 
Let $\Gamma \neq \emptyset$ be a finite set of \IP -formulae and $E \in \Gamma$.
Then, for each formula $A$ in \EP, Flagg and Friedman define the formula $
A_\Gamma^{(E)}$ of \IP~(simply denoted by $A^{(E)}$ when this is not ambiguous) as follows%
\footnote{Flagg and Friedman use $\neg_EA$ as an abbreviation for $(A\rightarrow E).$}%
, recursively on the structure%
\footnote{Please consider that Flagg and Friedman treat the sign of false, $\bot$, as an atomic formula.} %
of $A:$

\medskip

$\begin{array}{lr}
\medskip
A_\Gamma^{(E)} = \neg_E\neg_EA
, \hbox{if $A$ is atomic}; & (B_0\wedge B_1)_\Gamma^{(E)} = (B_{0~\Gamma}^{(E)}\wedge B_{1~\Gamma}^{(E)}); \\

\medskip 
(B_0\vee B_1)_\Gamma^{(E)} = \neg_E\neg_E(B_{0~\Gamma}^{(E)} \vee B_{1~\Gamma}^{(E)}); & \ \ \ (B_0\rightarrow B_1)_\Gamma^{(E)} = (B_{0~\Gamma}^{(E)}\rightarrow B_{1~\Gamma}^{(E)}); \\

(\Box B)^{(E)}_\Gamma = \neg_E \neg_E \bigwedge_{C \in \Gamma}B^{(C)}_\Gamma.

\end{array} $ 

\st
Their basic result is the following theorem (Th. 1.8 in \cite{FlaFri86}), 
which states the soundness of their translation in propositional calculus.

\begin{theorem}\label{th:ff-correct}
Let $A_1, \dots , A_n, A$ be \EP-formulae (possibly, $n=0$). If
\[
A_1, \dots , A_n \vdash_\EP A,
\]

\ni
then, for any finite set $\Gamma$ of \IP-formulae and for any $E \in \Gamma,$
\[
A_{1~\Gamma}^{(E)}, \dots , A_{n~\Gamma}^{(E)} \vdash_\IP A_\Gamma^{(E)}.
\]

\end{theorem}

\COMMENT
The authors do not give enough details are given for a straightforward proof that $(\cdot)^{(E)}_\Gamma$ is sound in general (i.e. irrespective of the choice of the finite set $\Gamma$ of intuitionistic formulae and the formula $E \in \Gamma$).
Moreover, the authors' suggestion to prove soundness by induction on length of derivations, if applied plainly, does not lead to a successful proof.
In fact, a plain and standard induction on length of derivations would presuppose that the generic $(\cdot)^{(E)}_\Gamma$ preserved admissibility of all epistemic inference rules; but unfortunately 
we found a counter-example to preservation of the necessitation rule, supported by an algebraic counter-model%
\footnote{%
For the intuitionistic algebraic semantics and its adequacy w.r.t. intuitionistic logic, please see e.g. \cite{RasSik63}.}. 
Thus, to optain a successful proof, we define a more structured application of induction (no more based on preservation of the necessitation rule) that succeeds in proving soundness in the propositional case.
\ENDCOMMENT

The authors suggest to prove Theorem \ref{th:ff-correct} by induction on the length of derivations, but do not give enough details for a straightforward proof.
In fact, a plain and standard induction on the length of derivations would presuppose that the generic $(\cdot)^{(E)}_\Gamma$ preserved the admissibility of all epistemic inference rules. 
Unfortunately that allows counter-examples to the preservation of the necessitation rule, as shown in Section \ref{inadmissibility} by means of an algebraic counter-model%
\footnote{%
For the intuitionistic algebraic semantics and its adequacy w.r.t. intuitionistic logic, please see e.g. \cite{RasSik63}.}. 
Thus, to optain a successful proof, we define in Section \ref{soundness_proof} a more structured application of induction which is not based on preservation of the necessitation rule and succeeds in proving soundness in the propositional case.
 
\COMMENT
Then, we survey the literature related to the F\&F-translation, from earlier works on Friedman's own {\em $A$-translation} to ones, by Fernandez and Inou\'e, that considered and disproved notable properties (other than soundness) relative to $f.$  
Finally, we envision a study on the predicate case.


\section{Preliminaries}\label{sec:prelim}
In this article we consider the Intuitionistic and the Epistemic Propositional calculus, denoted as \IP~and \EP~respectively, formulated as systems of natural deduction as in Prawitz \cite{Pra71}.

Notation and conventions will adhere as much as possible to Flagg and Friedman \cite{FlaFri86}. 

As it is common, we take $\bot$ (the symbol for \textit{falsity}) as a primitive zero-ary connective; however, following \cite{FlaFri86}, we consider it an atomic formula. 
As a result, 
we define the unary connective for negation ($\neg$) as follows:
\[
\neg A = (A\rightarrow \bot).
\]

\ni
In \EP, the symbol $\Box$ represents the modal operator for necessity (or {\em knowability}), as usual. 
We shall denote the sets of the {\em well-formed formulae} of the languages of \textbf{IP} and \textbf{EP} by $\calw_\IP$ and $\calw_\EP,$ respectively.
We shall use $p, q, r, ...$ as standing for propositional letters; $A, B, C, ...$ as standing for well-formed formulae (of a fixed calculus); $\Gamma, \Delta, \Theta, ...$ as denoting {\em finite} sets of formulae.

\IP~has the following primitive inference rules:

\st 
\begin{center}
$ (\wedge\:\textbf{I}) \ \ \frac{A_0 ~ A_1}{A_0\wedge A_1},  \ \ \ \ \ \ \ \ \ \ \ \ \ \ (\wedge\:\textbf{E})  \ \ \frac{A_0\wedge A_1}{A_i}, \ \ \ \ \ \ \ \ \ \ \ \ \ \ (\vee\:\textbf{I}) \ \ \frac{A_i}{A_0\vee A_1}, \ \ \ \ \ \ \ \ \ \ \ \ \ \ \ (\vee\:\textbf{E}) \ \ \frac{A_0\vee A_1 ~ \stackrel{[A_0]}{\stackrel{\vdots}{B}} ~ \stackrel{[A_1]}{\stackrel{\vdots}{B}}}{B},$
\end{center}

\begin{center}
$ (\rightarrow\:\textbf{I}) \ \ \frac{\stackrel{\stackrel{[A]}{\vdots}}{B}}{A\rightarrow B}, \ \ \ \ \ \ \ \ \ \ \ \ \ \ (\rightarrow\:\textbf{E}) \ \ \frac{A\rightarrow B ~  A}{B}, \ \ \ \ \ \ \ \ \ \ \ \ \ \ (\bot_\textbf{i}) \ \ \frac{\bot}{A}.$
\end{center}

\medskip

\EP~is gained from \IP~by omitting the rule $\bot_\textbf{i}$ and adding the following:

\st
\begin{center}
$ (\Box\:\textbf{I}) \ \ \frac{A}{\Box A}$ ~ if all open assumptions are of the form $\Box B,$
\end{center}

\begin{center}
$ (\Box\:\textbf{E}) \ \ \frac{\Box A}{A}, \ \ \ \ \ \ \ \ \ \ \ \ \ \ (\bot_\textbf{c}) \ \ \frac{\stackrel{\stackrel{[\neg A]}{\vdots}}{\bot}}{A}. $
\end{center}

\st
The $\Box$\textit{-introduction} rule, $\Box\:\textbf{I}$, is also known as \textit{necessitation}.%

\medskip
As usual, if $A_1, ... , A_n, B$ are formulae of $\calw_\IP$, we shall say that $B$ is deducible from $A_1, ... , A_n,$ and write $A_1, ... , A_n \vdash_\IP B$, if there is a derivation in \IP~that ends with $B$ and has $A_1, ... , A_n$ as open assumptions, and we shall say that $A$ and $B$ are (syntactically) equivalent, and write $A \vdash_\IP\dashv B,$ if $A \vdash_\IP B$ and $B \vdash_\IP A.$ Furthermore, we shall write $\Gamma, A \vdash_\IP B$ instead of $\Gamma \cup {A} \vdash_\IP B.$ Analogous notations hold for \EP. 

We consider now the \G~translation $T$ from $\calw_\IP$ into $\calw_\EP$ as presented, e.g., in \cite{RasSik63}.
For any formula $A$ in $\calw_\IP,$ the formula $T(A)$ in $\calw_\EP$ is defined by recursion on the structure of $A$ as follows:

\begin{itemize}
	\item $T(p) = \Box p,$ for each propositional letter $p;$

  \item $T(\bot) = \bot;$

  \item$T(A_0\wedge A_1) = T(A_0)\wedge T(A_1);$

  \item$T(A_0\vee A_1) = T(A_0)\vee T(A_1);$

  \item $T(A_0\rightarrow A_1) = \Box (T(A_0)\rightarrow T(A_1).$

\end{itemize}


As a result, we have $T(\neg A) = \Box\neg\ T(A).$ 

Let us remark the following important property of $T,$ called \textit{stability:} for each formula $A,$

\[
T(A) \vdash_\EP \dashv \Box T(A).
\]

\subsection{
Double relative negation} 
In order to illustrate how the F\&F-translation works at the propositional level, we adopt notation style and lemmata from \cite{FlaFri86}. 
From now on, let $A, B, C,$ and $E$ be arbitrary \IP -formulae. 
The following abbreviation is usefully introduced.

\begin{definition}\label{def:rel_neg}
\[
\neg_EA = (A\rightarrow E).
\]
\end{definition}

The resulting connective ($\neg_E)$ may be called $E$\textit{-negation} or \textit{relative negation w.r.t.} $E$, since negation shares some key properties with it.
In this regard, consider the following lemmata, drawn from Lemma 1.6 of \cite{FlaFri86}, 
except for Lemma \ref{le:contraposition}.

\begin{lemma}\label{le:double_neg} 
\[
A \vdash_\IP \neg_E\neg_EA.
\]
\end{lemma}

Notice that the opposite direction obviously does not hold in general.

\begin{lemma}\label{le:contraposition} 

\[
  A\rightarrow B \vdash_\IP \neg_EB\rightarrow \neg_EA; \ \
  A\rightarrow B \vdash_\IP \neg_E\neg_EA\rightarrow \neg_E\neg_EB
\]  

\end{lemma}

\begin{lemma}\label{le:triple_neg} 
\[
\neg_EA \vdash_\IP\dashv \neg_E\neg_E\neg_EA.
\]
\end{lemma}

\begin{lemma}\label{le:2_neg_intro}
The following, called {\rm double $E$-negation} rule, is admissible in \IP: 

\st
\[ (\neg_E\neg_E) \ \ \ \frac{\neg_E\neg_EA ~ ~ ~ \stackrel{[A]}{\stackrel{\vdots}{B}}}{\neg_E\neg_EB} \]
\end{lemma}

\begin{lemma}\label{le:2_neg_con} 
\[
\neg_E\neg_E(A\wedge B) \vdash_\IP\dashv (\neg_E\neg_EA\wedge\neg_E\neg_EB).
\]

\end{lemma}

\begin{lemma}\label{le:2_neg_dis} 
\[
\neg_E\neg_E(A\vee B) \vdash_\IP\dashv \neg_E\neg_E(\neg_E\neg_EA\vee \neg_E\neg_EB).
\]
\end{lemma}

\subsection{Useful lemmata}
In order to simplify what follows, here we prove four new lemmata.

\begin{lemma}\label{le:double_double}
\[
\neg_E\neg_E A \vdash_\IP \neg_E\neg_E\neg_C\neg_C A.
\]
\end{lemma}

\begin{proof}
This is an immediate consequence of Lemmata \ref{le:double_neg} and \ref{le:2_neg_intro}.
\end{proof}

\begin{lemma}\label{le:double_neg_imp}
\[
\neg_E\neg_E (A\rightarrow B) \vdash_\IP (\neg_E\neg_E A\rightarrow \neg_E\neg_E B).
\]
\end{lemma}

\begin{proof}
By $\rightarrow\textbf{E},$ we have $(A\rightarrow B), A \vdash_\IP B$; hence, for Lemma \ref{le:2_neg_intro}, we obtain $\neg_E\neg_E(A\rightarrow B), \neg_E\neg_EA \vdash_\IP \neg_EB$; finally, the result follows by $\rightarrow\textbf{I}.$
\end{proof}

\begin{lemma}\label{le:imp_double_neg}
\[
(\neg_E\neg_E A\rightarrow \neg_E\neg_E B) \vdash_\IP\dashv \neg_E\neg_E(\neg_E\neg_E A\rightarrow \neg_E\neg_E B).
\]
\end{lemma}

\begin{proof}
The 
left-to-right direction follows from Lemma \ref{le:double_neg}; the opposite direction is obtained from Lemma \ref{le:double_neg_imp} by substituting equivalent formulae thanks to Lemma \ref{le:triple_neg}.
\end{proof}

\begin{lemma}\label{!}
\[
(A\rightarrow \neg_E\neg_EB) \vdash_\IP\dashv (\neg_E\neg_EA\rightarrow \neg_E\neg_EB).
\]

\end{lemma}

\begin{proof}
As for the left-to-right direction, by $\rightarrow\textbf{E}$ we have $(A\rightarrow \neg_E\neg_EB), A \vdash_\IP \neg_E\neg_EB,$ therefore for Lemma \ref{le:2_neg_intro} $(A\rightarrow \neg_E\neg_EB), \neg_E\neg_EA \vdash_\IP \neg_E\neg_E\neg_E\neg_EB,$ hence by Lemma \ref{le:triple_neg} and the composition principle $(A\rightarrow \neg_E\neg_EB), \neg_E\neg_EA \vdash_\IP \neg_E\neg_EB;$ the result follows by $\rightarrow\textbf{I}.$ As for the opposite direction, we have by $\rightarrow\textbf{E}$ that $(\neg_E\neg_EA\rightarrow \neg_E\neg_EB), \neg_E\neg_EA \vdash_\IP \neg_E\neg_EB;$ but for Lemma \ref{le:double_neg} $A \vdash_\IP \neg_E\neg_EA,$ hence for composition $(\neg_E\neg_EA\rightarrow \neg_E\neg_EB), A \vdash_\IP \neg_E\neg_EB;$ the result follows by $\rightarrow\textbf{I}$ again.

\end{proof}

We can now show the translation $f$ from $\calw_\EP$ into $\calw_\IP$ (Definition 1.7 of \cite{FlaFri86}) and discuss its properties.

\section{The F\&F-translation} 

Flagg and Friedman introduced, at the propositional level, the following translation (the above-said F\&F-translation).
 
\begin{definition}\label{def:ff-trans}
Let $\Gamma \neq \emptyset$ be a finite set of \IP -formulae and let $E \in \Gamma$. 
Then, for each formula $A$ in \EP, the formula $
f_\Gamma^{(E)}(A) = A_\Gamma^{(E)}$ of \IP~(simply denoted by $A^{(E)}$ or even by $f(A),$ when this is not ambiguous) is defined as follows, recursively on the structure of $A:$

\begin{itemize}
\item $A_\Gamma^{(E)} = \neg_E\neg_EA,$\ if $A$ is atomic;

\medskip
\item $(B\wedge C)_\Gamma^{(E)} = (B_\Gamma^{(E)}\wedge C_\Gamma^{(E)});$

\medskip 
\item $(B\vee C)_\Gamma^{(E)} = \neg_E\neg_E(B_\Gamma^{(E)} \vee C_\Gamma^{(E)});$

\medskip
\item $(B\rightarrow C)_\Gamma^{(E)} = (B_\Gamma^{(E)}\rightarrow C_\Gamma^{(E)});$

\medskip
\item $(\Box B)^{(E)}_\Gamma = \neg_E \neg_E \bigwedge_{C \in \Gamma}B^{(C)}_\Gamma.$
\end{itemize}

\end{definition}

\noindent
For the sake of completeness, we also report the two additional clauses needed for {\em predicate} formulas, which are not in this work.

\begin{itemize}
 \item $(\forall xA)^{(E)}_\Gamma = \forall xA^{(E)}_\Gamma;$

 \medskip
 \item $(\exists xA)^{(E)}_\Gamma = \neg_E\neg_E\exists xA^{(E)}_\Gamma.$
\end{itemize}

\noindent
The translation has two immediate consequences, which are used to simplify formulae in what follows: 

\begin{itemize}
\item $\bot^{(E)}_\Gamma \vdash_\IP\dashv E$;

\medskip
\item $(\neg A)^{(E)}_\Gamma \vdash_\IP\dashv \neg_E A^{(E)}_\Gamma.$
\end{itemize}

Even though they are missing from Flagg and Friedman's article, consider now some simple examples of how $f^{(E)}_\Gamma,$ works. 
Let $p,$ $q,$ and $r$ be propositional formulae and let $\Gamma = \{E, C\} \subseteq \calw_\IP.$

\begin{example} (Translating a formula with $\vee$ and $\wedge$)

\st
$\begin{array} {rl}
 (p\wedge(q\vee r))^{(E)}_\Gamma\ = \ & (\neg_E\neg_Ep\wedge\neg_E\neg_E(\neg_E\neg_Eq\vee \neg_E\neg_Er))\\
  & \vdash_\IP\dashv \ (\neg_E\neg_Ep\wedge\neg_E\neg_E(q\vee r)) \hspace{3,4cm} \hbox{(for Lemma \ref{le:2_neg_dis})}\\
  & \vdash_\IP\dashv \ \neg_E\neg_E(p\wedge (q\vee r)) \hspace{4,5cm} \hbox{(for Lemma \ref{le:2_neg_con})}
\end{array}$
\end{example}
\begin{example} (Translating a formula with $\Box$ but no $\rightarrow$)

\st
$\begin{array} {rlr}
(\Box p)^{(E)}_\Gamma \ = \ & \neg_E\neg_E (\neg_E\neg_Ep \wedge\neg_C\neg_Cp) & \\
& \vdash_\IP\dashv \ (\neg_E\neg_E\neg_E\neg_Ep \wedge\neg_E\neg_E\neg_C\neg_Cp) & \ \ \ \ \ \ \ \ \ \ \ \ \ \ \ \ \ \ \ \ \ \ \ \hbox{(for Lemma \ref{le:2_neg_con})}\\
& \vdash_\IP\dashv \ (\neg_E\neg_Ep \wedge\neg_E\neg_E\neg_C\neg_Cp) & \ \ \ \ \ \ \ \ \ \ \ \ \ \ \ \ \ \ \ \ \ \ \ \ \hbox{(for Lemma \ref{le:triple_neg})} \\
&  \vdash_\IP\dashv \ \neg_E\neg_Ep & \ \ \ \ \ \ \ \ \ \ \ \ \ \ \ \ \ \ \ \ \ \ \ \hbox{(for Lemma \ref{le:double_double})}
\end{array}$
\end{example}
\begin{example} (Translating a formula with at least one $\rightarrow$)

\st
$\begin{array} {r l}
((p\rightarrow q)\vee r)^{(E)}_\Gamma \ = \ & \neg_E\neg_E((\neg_E\neg_Ep\rightarrow \neg_E\neg_Eq)\vee \neg_E\neg_Er) \\
& \vdash_\IP\dashv \ \neg_E\neg_E(\neg_E\neg_E(\neg_E\neg_Ep\rightarrow \neg_E\neg_Eq)\vee \neg_E\neg_Er) \ \ \hbox{(for Lemma \ref{le:imp_double_neg})} \\
& \vdash_\IP\dashv \ \neg_E\neg_E((\neg_E\neg_Ep\rightarrow \neg_E\neg_Eq)\vee r) \ \ \ \ \ \ \ \ \ \ \ \ \ \ \ \ \ \hbox{(for Lemma \ref{le:2_neg_dis})}\\
& \vdash_\IP\dashv \ \neg_E\neg_E((p\rightarrow \neg_E\neg_Eq)\vee r) \ \ \ \ \ \ \ \ \ \ \ \ \ \ \ \ \ \ \ \ \ \ \ \ \ \hbox{(for Lemma \ref{!})}
\end{array}$
\end{example}



\medskip\noindent
Finally, note how F\&F-translated formulae enjoy the following important property.

\begin{lemma}\label{le:double_neg_elim}
\[
\neg_E \neg_E A^{(E)}_\Gamma \vdash_\IP\dashv A^{(E)}_\Gamma.
\]

\end{lemma}

This result appears without proof in \cite{FlaFri86}; here it can be proved straightforwardly by induction on the structure of $A,$ thanks to Lemmata \ref{le:triple_neg} and \ref{le:2_neg_con} by Flagg and Friedman and our Lemma \ref{le:imp_double_neg}.

\subsection{The soundness issue}
Flagg and Friedman's main result is the following theorem (Th. 1.8 in \cite{FlaFri86}), 
which states the soundness of their translation in propositional calculus.

\begin{theorem}\label{th:ff-correct}
Let $A_1, \dots , A_n, B \in \calw_\EP.$ If
\[
A_1, \dots , A_n \vdash_\EP A,
\]

\ni
then, for any finite $\Gamma \subseteq \calw_\IP$ and for any $E \in \Gamma,$
\[
A_{1~\Gamma}^{(E)}, \dots , A_{n~\Gamma}^{(E)} \vdash_\IP A_\Gamma^{(E)}.
\]

\end{theorem}

This theorem can be proved, as Flagg and Friedman suggest, by induction on length of derivations; however, they do not prove it explicitly.
We find that the proof is not exactly trivial and their suggestion might easily be misunderstood.


\ENDCOMMENT

\section{Inadmissibility of the necessitation rule}\label{inadmissibility}
If we were to prove Th. \ref{th:ff-correct} above by a plain induction on length of derivation, then, according to a standard procedure, we would expect that, for any finite $\Gamma \subseteq \calw_\IP$ and any $E \in \Gamma,$ the translation $(\cdot)^{(E)}_\Gamma$ would transform 
any \EP -primitive inference rule into an \IP -admissible inference rule. 
It can be easily proved for all \EP -primitive rules but for the $\Box$-introduction rule (denoted by $\Box \:\textbf{I}$): 

\begin{center}
$ \Box\:\textbf{I} \ \ \frac{A}{\Box A}$ ~ if all open assumptions are of the form $\Box B.$
\end{center}

\COMMENT
Consider, e.g., the $\vee$-introduction rule for \EP. 
Here is a proof of the admissibility of $\vee \textbf{I}^{(E)}_\Gamma$ for arbitrary $E$ and $\Gamma.$

\begin{enumerate}
	\item $\Delta^{(E)} \vdash_\IP \neg_E\neg_E(A^{(E)}\vee B^{(E)}$~~~~~~(for hypothesis).
	\item $\Phi^{(E)}, A^{(E)} \vdash_\IP C^{(E)}$~~~~~~(for hypothesis).
	\item $\Sigma^{(E)}, B^{(E)} \vdash_\IP C^{(E)}$~~~~~~(for hypothesis).
	\item $\Phi^{(E)}, \Sigma^{(E)}, A^{(E)}\vee B^{(E)} \vdash_\IP C^{(E)}$~~~~~~(by $\vee\textbf{E}$ from \textit{2} and \textit{3}).
	\item $\Phi^{(E)}, \Sigma^{(E)}, \neg_E\neg_E (A^{(E)}\vee B^{(E)}) \vdash_\IP \neg_E\neg_E C^{(E)}$~~~~~~(for Lemma \ref{le:2_neg_intro}).
	\item $\Delta^{(E)}, \Phi^{(E)}, \neg_E\neg_E(A^{(E)}\vee B^{(E)}) \vdash_\IP C^{(E)}$~~~~~~(for Lemma \ref{le:double_neg} and composition).
	\item $\Delta^{(E)}, \Phi^{(E)}, \Sigma^{(E)}\vdash_\IP C^{(E)}$~~~~~~(by composition from \textit{1} and \textit{6}).
\end{enumerate}

\ENDCOMMENT

\ni
In fact, there is at least one counter-example to the admissibility of every $(\Box \:\textbf{I})_\Gamma^{(E)}$ in \IP, 
as we show now.

\begin{theorem}\label{pr:counter_ex}
Let $B$ and $C$ be arbitrary atomic formulae (including $\bot$) and let $E$ be a propositional letter distinct from $B$ and $C$; moreover, let $\Gamma = \{C, E\}.$ 
If $A = (E\rightarrow B),$ then 

\[
\vdash_\IP A_\Gamma^{(E)},
\]

\ni
while

\[
\not\vdash_\IP (\Box A)_\Gamma^{(E)}.
\]
\end{theorem}

\begin{proof}
Clearly $A_\Gamma^{(E)} \vdash_\IP\dashv E\rightarrow \neg_E\neg_EB,$ 
hence $A_\Gamma^{(E)}$ is an \IP -theorem. 
Let us then prove that, in contrast, $(\Box A)_\Gamma^{(E)}$ is not. 
Since $(\Box A)_\Gamma^{(E)} = \neg_E\neg_E(A^{(C)}\wedge A^{(E)}),$ thanks 
to Lemma 1.6/\textit{(v)} of \cite{FlaFri86} and the $\wedge$-elimination rule 
it suffices to prove that $\not\vdash_\IP \neg_E\neg_E A^{(C)}.$ 
We shall show algebraically that $\neg_E\neg_E A^{(C)}$ is not intuitionistically valid, namely:
\[
\not\models_\IP \neg_E\neg_E(\neg_C\neg_C E\rightarrow \neg_C\neg_C B);
\]

\ni
thanks to the soundness of \IP\ wrt. the algebraic semantics, we obtain the result.

In fact, let $\mathfrak{H} = (H, \leq)$ be a totally-ordered Heyting algebra%
\footnote{%
A Heyting algebra is a bounded lattice $\mathfrak{G} = (G; \preceq; \cup, \cap; 0, 1)$ (where $G$ is the subjacent set of $\mathfrak{G},$ the relation $\preceq$ is the characteristic order of the lattice, $\cup$ and $\cap$ are the usual {\em supremum} and {\em infimum} operators respectively, 0 and 1 the top and the bottom of $\mathfrak{G}$ respectively) such that, for any $x, y \in G,$ there exists the \textit{relative pseudo-complement} of $x$ w.r.t. $y$, denoted by $x\triangleright y,$ i.e. there exists a $z \in G$ such that, for any $w \in G$, we have: $w \preceq z\ \Longleftrightarrow \ w\cap x \preceq y$.
} 
(i.e. a chain) such that $b, c, e \in H$ and $0 \leq b \leq c < e < 1$. 
Now, let $v$ be a valuation in $H$ (i.e. a function from the set of propositional letters to $H$) such that, denoting by $v_\mathfrak{H}$ the standard extension of $v$ to the \IP-formulae w.r.t. $\mathfrak{H},$ we obtain $v_\mathfrak{H}(B) = b$, $v_\mathfrak{H}(C) = c,$ and $v(E) = e$. 
Then 

\[ 
v_\mathfrak{H}(\neg_E\neg_EA^{(C)}) = (e\triangleright c\triangleright c)\triangleright(b\triangleright c\triangleright c)\triangleright e\triangleright e = e \neq 1.
\]
 
\ni
Hence $v_\mathfrak{H}(\neg_E\neg_EA^{(C)}) \neq 1$. 
As a result, $\not\models_\IP \neg_E\neg_EA^{(C)}$. 
\end{proof}

\section{A non-trivial proof of soundness}\label{soundness_proof}
We shall propose a new, explicit proof for Th. \ref{th:ff-correct}, with a focus on the 
treatment of the necesitation rule.

\begin{proof}
Let $\delta$ be an $\EP$ -derivation, in one arbitrary step of which a formula $A$ is asserted under the assumptions $A_1, \dots , A_n$. Also, let $\Gamma$ be an arbitrary finite set of $\IP$-formulae. 
We suppose, as induction hypothesis, that for each $C \in \Gamma$ the asserted formula in any previous step of $\delta$, once translated w.r.t. $C$ and $\Gamma,$ is provable in \IP~from the translations of the open assumptions concerning that step of $\delta$.

\COMMENT
Let then $\delta$ be an \EP-derivation, of length $l$ (let us imagine here natural deduction in a sequential presentation), such that the last step of $\delta$ consists in the assertion of $B$ under the assumptions $A_1, \dots , A_n.$
Next, choose arbitrarily $k \leq l$ and a finite $\Gamma \subseteq \calw_\IP.$

We claim that, for all $C \in \Gamma,$ the formula asserted in the $k$-th step of $\delta,$ once it gets translated w.r.t. $C$ and $\Gamma,$ turns out to be provable in \IP~from the the translations of the open assumption of that step of $\delta,$ under the induction hypothesis that, for all $C \in \Gamma,$ all steps in $delta$ preceding that 
are such.

Consider the case that the $k$-th step in $\delta$ is obtained from an earlier step, say the $i$-th one with $i < k,$ through the rule $\Box\:\textbf{I}$. 
Let $\Delta$ and $A$ be the set of open assumptions and the asserted formula at the $i$-th step of $\delta,$ respectively.
\ENDCOMMENT

Consider the case that the assertion of $A$ is obtained--through the $\Box\:{\rm\bf I}$ rule--from an earlier step of $\delta.$
Let then $B_1, \dots, B_n, B$ be such that $A_1 = \Box B_1, \dots, A_n = \Box B, A = \Box B.$
We claim that, for any $E \in \Gamma,$

\[
\neg_E\neg_E\bigwedge_{C \in \Gamma}B^{(C)}_{1~\Gamma}, \dots, \neg_E\neg_E\bigwedge_{C \in \Gamma}B^{(C)}_{n~\Gamma} \vdash_\IP \neg_E\neg_E\bigwedge_{C \in \Gamma}B^{(C)}_\Gamma.
\]

\ni
To prove it, thanks to Lemma 
1.6/ \textit{iv} of \cite{FlaFri86} it suffices to prove, by fixing an arbitrary formula $D \in \Gamma,$ 
that $\neg_E\neg_EB^{(D)}_\Gamma$ is provable from the above-mentioned assumptions. 

\ni
\COMMENT
Suppose now that $\Delta = \{\Box B_1, \dots ,\Box B_m\}$ 
with $m \geq 1$ (as the case $\Delta = \emptyset$ is straightforward).
\ENDCOMMENT

\ni
By the induction hypothesis 
we have (let us omit $\Gamma$ from now on) 

\[
\neg_D\neg_D\bigwedge_{C}B_1^{(C)}, \dots , \neg_D\neg_D\bigwedge_{C}B_n^{(C)} \vdash_\IP B^{(D)},
\]

\ni
from which, by 
Lemma 1.6/\textit{(vii)} we obtain:

\[
\neg_E\neg_E\neg_D\neg_D\bigwedge_{C}B_1^{(C)}, \dots , \neg_E\neg_E\neg_D\neg_D\bigwedge_{C}B_n^{(C)} \vdash_\IP \neg_E\neg_EB^{(D)}.
\]

\ni 
Now, for all $i$ 
such that $1 \leq i \leq n$, 
by Lemma 1.6/\textit{(i)}, \textit{(vii)} we have the following:

\[
\neg_E\neg_E\bigwedge_{C}B_i^{(C)} \vdash_\IP \neg_E\neg_E\neg_D\neg_D\bigwedge_{C}B_i^{(C)}.
\]

\ni
Finally, by composition we get to:

\[
\neg_E\neg_E\bigwedge_{C}B_1^{(C)}, \dots , \neg_E\neg_E\bigwedge_{C}B_n^{(C)} \vdash_\IP \neg_E\neg_EB^{(D)}.
\] 

\hfill$\Box$

\end{proof}

\COMMENT
The following corollary (listed as Th. 1.9 in \cite{FlaFri86}) asserts that $f^{(E)}_\Gamma$ is sound for every $\Gamma$ and every $E$ (notice how the corollary is in fact equivalent to Th. \ref{th:ff-correct} above).

\begin{corollary}\label{co:ff-correct}
For each $A \in \calw_\EP$, if
\[
\vdash_\EP A,
\]

\ni 
then, for any finite $\Gamma \in \calw_\IP$ and any $E \in \Gamma$,
\[
\vdash_\IP A^{(E)}_\Gamma.
\]
\end{corollary}

\section{Relation to literature}\label{sec:lit}
Among the first to study the epistemic extensions of classical theories, Shapiro \cite{Sha85b} introduced Epistemic Arithmetic \EA~and observed that the \G~translation $T$ is sound w.r.t. \HA~and \EA.
At about the same time, Goodman \cite{Goo84a} proved the faithfulness of $T$ for arithmetic via a proof-theoretical method; that was also proved ---syntactically--- for versions of set theory and type theory, albeit with a different proof method for each formal system.
The goal of Flagg and Friedman \cite{FlaFri86} was to find a uniform method, 
precisely one that was based on the soundness of their translation \textit{f.}

\subsection{The $A$-translation and its connection with the F\&F-translation}
The F\&F-translation is closely related to Friedman's better-known $A${\em-translation} $(\cdot)_A$ -- or {\em disjunctive} translation w.r.t. an intuitionistic formula $A:$ it is an internal operation on the set of intuitionistic formulae that he had previously proposed in \cite{Fri78} in order to prove that Markov's inference rule is admissible in several intuitionistic theories, in particular \HA.

As Flagg and Friedman pointed out%
\footnote{%
Please see \cite{FlaFri86}, p. 60.
},
the $A$-translation can be described as a special case of their translation $f^{(E)}_\Gamma.$ 
Strictly speaking, $f^{(E)}_\Gamma$ is (modulo logical equivalence) an extension to epistemic logic -- w.r.t. the finite set $\Gamma$ of intuitionistic formulae -- of the composition of the \G-Gentzen {\em double-negation} translation from classical into intuitionistic logic with the $A$-translation obtained by setting the formula $E$ as $A.$ 
This fact can easily be proved by induction on a generic argument of the \G-Gentzen translation. 
Notice, moreover, how the $A$-translation is classically equivalent to the F\&F-translation restricted to classical formulae and, in general, the double negative negation w.r.t. $E$ is classically equivalent to the simple disjunction with $E.$


The F\&F-translation, in its propositional version, has been treated in literature w.r.t. its properties, or lack of them. 
Some properties have been disproved, particularly by Fernandez \cite{Fer06} and by Inou\'e \cite{Ino92}, who have produced the counter-examples that we shall discuss next.
 
\subsection{Unfaithfulness 
of the F\&F-translation in literature} 
We have shown above a proof of soundness for every propositional translation $f^{(E)}_\Gamma.$ 
One may wonder if, moreover, all $f^{(E)}_\Gamma$s are faithful.
It remains relatively easy to see that this is not the case; the following counter-example is borrowed from Fernandez%
\footnote{
Fernandez \cite{Fer06} then proposed an alternative polynomial translation from \EP~to \IP, which he proved to be faithful.
}. 

Fix an arbitrary theorem of \IP, e.g. $\top = \bot\rightarrow \bot,$ as $E,$ and an arbitrary extension of $\{E\}$ as $\Gamma.$
Now, let $p$ be an arbitrary propositional letter; then $p^{(E)} = 
((p\rightarrow \top)\rightarrow \top).$
Hence we have $\vdash_\IP p^{(E)}$, whereas obviously $\not\vdash_\EP p;$ from that we can infer that $f^{(E)}_\Gamma$ is unfaithful.

It could also be interesting to check whether our conjecture about the faithfulness of the F\&F-translation may hold {\em in a weaker form:} e. g., would the opposite of Corollary \ref{co:ff-correct} hold? 
In other words, for any $A \in \calw_\EP$ can we prove that
\[ \vdash_\IP A^{(E)}_\Gamma~\hbox{for each finite}~\Gamma~\hbox{and each}~E\in \Gamma\ \Longrightarrow \ \vdash_\EP A? \]
 
Even in this weaker form, the conjecture does not hold, as Inou\'e proved by means of the following counter-example \cite{Ino92}. Let

\[
A = p\rightarrow \Box p.
\]

\ni 
Thanks to lemmata proved earlier on, in particular Lemma \ref{le:double_double}, it remains easy to prove that $\vdash_\IP A^{(E)}_\Gamma,$ irrespective of $\Gamma$ and $E,$ while unfortunately $\not\vdash_\EP A$.
Inou\'e's result prompts us to argue that no translation $f^{(E)}_\Gamma$ may be faithful.

\section{Conclusions and future research}\label{sec:conc}
The method proposed by Flagg and Friedman is appealing, as it would prove that the \G~translation $T$ is faithful for various intuitionistic formal systems mapped onto the corresponding epistemic ones. 
The two authors aimed at finding a counter-translation of the epistemic image of an intuitionistic theory back into the intuitionistic theory itself, and to prove that it was at the same time sound and inverse (modulo logical equivalences) to \textit{T.} 
To do so, they had had to rely on the soundness of the F\&F-translation $f$ (relative to any intuitionistic formula $E$ and any finite set $\Gamma).$ 


We have examined their translation for the propositional case and found out the inadmissibility of the (translated) necessitation rule. 
This result prompted us to propose a detailed proof of the soundness of \textit{f,} as it was left implicit in their article \cite{FlaFri86}. 
Our results suggest that also the soundness proof for the F\&F-translation of {\em predicate} formulae should be reconsidered, in order to ensure that Flagg and Friedman's method for proving the faithfulness of the \G~translation is suitable for predicate logic and theories.


\ENDCOMMENT


\bibliographystyle{elsarticle-num}
\bibliography{../materials}
\end{document}